\documentclass[runningheads]{llncs}
\pdfoutput=1

\usepackage[utf8]{inputenc}
\usepackage{amssymb}
\usepackage{amsmath}
\usepackage{mathrsfs}
\usepackage{tikz}
\usetikzlibrary{shapes,arrows,chains,shadows}
\usepackage{mathpartir} 

\usepackage[hyperfootnotes=false]{hyperref}

\usepackage{todonotes}
\renewcommand\todo[1]{}
\usepackage{multicol}

\newcommand\blfootnote[1]{%
  \begingroup%
  \renewcommand\thefootnote{}\footnote{#1}%
  \addtocounter{footnote}{-1}%
  \endgroup%
}

\newcommand{\N}[1]{\ensuremath{{{#1}}}}
\newcommand{\M}[1]{\ensuremath{\mathrm{#1}}}

\usepackage[neveradjust,flushleft]{paralist}
\renewenvironment{enumerate}{\compactenum}{\endcompactenum}

\usepackage[countsame,nosections]{coqtheorem}

\numberwithin{lemma}{section}

\setBaseUrl{{https://www.ps.uni-saarland.de/extras/cbvlcm2/doc/LM.}}

\usepackage[location=appendix, manual=true]{moveproofs}
\crefname{factAuxCoq}{fact}{facts}
\crefname{factAux}{fact}{facts}
\crefname{theoremAuxCoq}{theorem}{theorems}
\crefname{lemmaAuxCoq}{lemma}{lemmas}
\crefname{corollaryAux}{corollary}{corollaries}

\usepackage{macros}

\renewcommand\N[1]{{#1}}
\newcommand{\coqemph}[2][]{\coqlink[#1]{\emph{#2}}}

\renewcommand{\M}[1]{\ensuremath{\mathsf{#1}}}
\renewcommand\phi\varphi

\newcommand\Ter{\M{Ter}}
\newcommand\Pro{\M{Pro}}
\newcommand\Clo{\M{Clo}}
\newcommand\Code{\M{Code}}
\newcommand\PA{\M{PA}}
\newcommand\Com{\M{Com}}
\newcommand\Heap{\M{Heap}}
\newcommand\HA{\M{HA}}
\newcommand\HC{\M{HC}}
\newcommand\HE{\M{HE}}
\newcommand\eval{\rhd}
\newcommand\red{\succ}
\newcommand\tred{\succ_\tau}
\newcommand\bred{\succ_\beta}
\newcommand\ret{\M{ret}}
\newcommand\var{\M{var}}
\newcommand\app{\M{app}}
\newcommand\lamb{\M{lam}}
\newcommand\bnd{\mathrel{<}}

\usepackage{tikz-cd}

\begin{document}
\title{Formal Small-step Verification of a\\
  Call-by-value Lambda Calculus Machine}
\author{Fabian Kunze, Gert Smolka, Yannick Forster}
\institute{Saarland University, Saarbrücken, Germany
\email{\{kunze,smolka,forster\}@ps.uni-saarland.de}}

\maketitle

\begin{abstract}
  \noindent
  We formally verify an abstract machine
  for a call-by-value $\lambda$-calculus
  with de Bruijn terms, simple substitution, 
  and small-step semantics.
  We follow a stepwise refinement approach
  starting with a naive stack machine with substitution. 
  We then refine to a machine with closures,
  and finally to a machine with a heap
  providing structure sharing for closures.
  We prove the correctness of
  the three refinement steps
  with compositional small-step bottom-up simulations.
  There is an accompanying Coq development
  verifying all results.
\end{abstract}

\enlargethispage{2cm}
\blfootnote{\hspace{-5.5mm}\copyright\ Springer Nature Switzerland AG 2018\\
S. Ryu.\ (Ed): APLAS 2018, LNCS 11275, pp.\ 264-283, 2018.\\
The final authenticated publication is available online at: \url{https://doi.org/10.1007/978-3-030-02768-1_15}}\vspace{-5mm}

\section{Introduction}

The call-by-value $\lambda$-calculus
is a minimal functional programming language
that can express recursive functions
and inductive data types.
Forster and Smolka~\cite{Forster17} employ
the call-by-value $\lambda$-calculus as
the basis for a 
constructive theory of computation
and formally verify elaborate programs
such as step-indexed self-interpreters.
Dal Lago and Martini~\cite{DalLagoMartini08} show that
Turing machines and the 
call-by-value \text{$\lambda$-calculus}
can simulate each other 
within a polynomial time overhead
(under a certain cost model).
Landin's SECD machine implements
the call-by-value \text{$\lambda$-calculus}
with closures eliminating the need for substitution~\cite{Landin64,Plotkin75}.

In this paper we consider the 
call-by-value $\lambda$-calculus L
from~\cite{Forster17}.
L~comes with de Bruijn terms
and simple substitution, and restricts
$\beta$-reduction to terms of 
the form $(\lambda s)(\lambda t)$
that do not appear within abstractions.
This is in contrast to
Plotkin's call-by-value $\lambda$-calculus~\cite{Plotkin75},
which employs terms with named argument variables 
and substitution with renaming,
and $\beta$-reduces terms of the forms 
$(\lam xs)(\lam yt)$ and $(\lam xs)y$.
L and Plotkin's calculus agree for closed terms,
which suffice for functional computation.

The subject of this paper is the formal verification
of an abstract machine for~L with
closures and structure sharing.
Our machine differs from the SECD machine
in that it operates on programs rather than terms,
has two flat stacks rather than one stack of frames,
and provides structure sharing through a heap.
Our goal was to come up with 
a transparent machine design providing for 
an elegant formal verification.
We reach this goal with 
a stepwise refinement approach
starting with a naive stack machine
with programs and substitution. 
We then refine to a machine with closures,
and finally to a machine with a heap.
As it comes to difficulty of verification,
the refinement to the naive stack machine
is by far the most substantial.


We prove the correctness of
the three refinement steps
with compositional small-step bottom-up simulations
(i.e., L is above the machines and simulates machine transitions)\todo{Mention high/low-level?}.
While L has only $\beta$-steps,
our machines have $\beta$- and \text{$\tau$-steps}.
$L$~simulates a machine by following $\beta$-steps
and ignoring $\tau$-steps,
and a machine simulates a lower-level machine 
by following $\beta$-steps with $\beta$-steps
and \text{$\tau$-steps} with $\tau$-steps.
To obtain bisimulations, 
we require progress conditions:
Reducibility must propagate downwards and
machines must stop after finitely many $\tau$-steps.

The first verification step establishes
the naive stack machine as a correct implementation of~L,
the second verification step establishes 
the closure machine as a correct implementation of 
the naive stack machine,
and the third verification step establishes
the heap machine as a correct implementation of
the closure machine.
The second and third verification step
are relatively straightforward since
they establish strict simulations (no silent steps).
Strict simulations suffice since 
the programs of the naive stack machine 
already provide the right granularity
for the structure sharing heap machine.

The entire development is formalised with 
the Coq proof assistant~\cite{Coq}.
Coq's type theory provides an ideal foundation
for the various inductive constructions needed for 
the specification and verification of the machines.  
All reasoning is naturally constructive.
In the paper we don't show Coq code but use
mathematical notation and language throughout.
While familiarity with constructive type theory
is helpful for reading the paper,
technical knowledge of Coq is not required.
For the expert and the curious reader,
the definitions and theorems in the paper 
are hyperlinked with their formalisations 
in an HTML rendering of the Coq  development.
The Coq formalisation is available at 
\url{https://www.ps.uni-saarland.de/extras/cbvlcm2/}.
\todo{coq-links}
\todo{check coq}
\subsection*{Related Work}

We review work concerning the verification
of abstract machines for call-by-value 
$\lambda$-calculus.

Plotkin~\cite{Plotkin75} presents
the first formalisation and verification 
of Landin's SECD machine~\cite{Landin64}.
He considers terms and closures with named
variables and proves that his machine
computes normal forms of closed terms using 
a step-indexed evaluation semantics for terms
and top-down arguments (from $\lambda$-calculus to machine).
He shows that failure of term evaluation 
for a given bound
entails failure of machine execution for this bound.
Plotkin does not prove his substitution lemmas.
Ramsdell~\cite{Ramsdell99} reports on a formalisation of
a Plotkin-style verification of
an SECD machine optimising tail calls
using the Boyer-Moore theorem prover.
Ramsdell employs de Bruijn terms 
and de Bruijn substitution.

Felleisen and Friedman~\cite{Felleisen86} study 
 Plotkin's call-by-value $\lambda$-calculus 
extended with control operators like J and call/cc.
They prove correctness properties relating
abstract machines, small-step reduction systems,
and algebraic theories.
Like Plotkin, they use terms and closures with named variables.

Rittri~\cite{rittri1988} seems to be the first 
who verifies an abstract machine for 
a call-by-value $\lambda$-calculus
using a small-step bottom up simulation.
Rittri's work is also similar to ours 
in that he starts from a $\lambda$-calculus 
with simple substitution reducing closed terms,
and in that his machine uses a control and an argument stack.
Rittri gives detailed informal proofs
using terms with named variables.
He does not consider a naive intermediate machine 
nor a heap realisation. 

Hardin et al.~\cite{Hardin98} 
verify several abstract machines
with respect to a fine-grained $\lambda$-calculus 
with de Bruijn terms and explicit substitution primitives.  
Like us, they simulate
machine steps with reduction steps of the calculus
and disallow infinitely many consecutive silent steps.
They consider 
the Krivine machine~\cite{Cregut1990} (call-by-name),
the SECD machine~\cite{Landin64,Plotkin75} (call-by-value),
Cardelli's FAM~\cite{Cardelli84} (call-by-value), and
the categorical abstract machine~\cite{Cousineau87}
(call-by-value) .

Accattoli et al.~\cite{Accattoli14} 
verify several abstract machines for
the linear substitution calculus with
explicit substitution primitives.
They simulate machine steps with reduction steps
of the calculus and model internal steps of the calculus
with a structural congruence.
They employ a global environment acting as heap.
Among other machines, 
they  verify a simplified variant of 
the ZINC machine \cite{leroy90}.

Leroy~\cite{Leroy2009,Leroy16} verifies 
the Modern SECD machine for
call-by-value $\lambda$-calculus
specified with de Bruijn terms and
an environment-based evaluation semantics in Coq.
The modern SECD machine has programs and a single stack.
Leroy's semantic setup is such that 
neither substitution nor 
small-step reduction of terms
have to be considered.
He uses top-down arguments and compiles
terms into machine states.
Using coinductive divergence predicates,
Leroy shows that the machine diverges
on states obtained from diverging terms.
Leroy's proofs are pleasantly straightforward.

Danvy and Nielsen~\cite{Danvy04refocusing} 
introduce the refocusing technique, 
a general procedure transforming
small-step reduction systems defined with evaluation contexts 
into abstract machines operating on the same syntax.  
Biernacka and Danvy~\cite{Biernacka07} 
extend refocusing and obtain
environment-based abstract machines. 
This yields a framework where the
derived machines are provably correct 
with respect to small-step bisimulation. 
Biernacka et al.~\cite{Biernacka17} formalise 
a generalisation of the framework in Coq.

Swierstra~\cite{Swierstra12} formally verifies the correctness
of a Krivine machine for simply typed $\lambda$-calculus
in the dependently typed programming language Agda.
Also following Biernacka and Danvy~\cite{Biernacka07},
Swierstra does this
by showing the correctness of a
Krivine-style evaluator for
an iterative and environment-based head reduction evaluator.
This way substitution does not appear.
Swierstra's dependently typed constructions
also provide normalisation proofs for
simply typed $\lambda$-calculus.
Swierstra's approach will not work 
for untyped $\lambda$-calculus.



\subsection*{Contribution of the Paper}

We see the main contribution of the paper
in the principled formal verification 
of a heap machine for a call-by-value $\lambda$-calculus
using a small-step bottom-up simulation.
A small-step bottom-up verification
is semantically more \todo{clearify} informative than the usual
evaluation-based top-down verification 
in that it maps every reachable  machine state to a term of L.
The entire Coq development consists of
500 lines of proof plus 750 lines of specification.
The decomposition of the verification 
in three refinement steps provides for 
transparency and reusability.
The use of the naive stack machine 
as an intermediate machine appears to be new.
We also think that our simple formalisation
of structure sharing with code and heap
is of interest.

We envision a formal proof showing that
Turing machines can simulate L with
polynomial overhead in time and 
constant overhead in space
(under a suitable cost model)~\cite{LOLA}.
The verifications in this paper
are one step into this direction.

\subsection*{Plan of the Paper}

After some preliminaries 
fixing basic notions in Coq's type theory,
we specify the call-by-value $\lambda$-calculus L
and present our abstract framework 
for machines and refinements.
We then introduce programs and program substitution
and prove a substitution lemma.
Next we specify and verify the naive stack machine for L.
This is the most complex refinement step as it comes to proofs.
Next we specify the closure machine
and verify that it is an implementation 
of the naive stack machine and hence of L (by compositionality).
Finally, we define abstractions for codes and heaps
and verify that the heap machine is an implementation
of the closure machine and hence of~L.
\todo{This was out for the initial submisission,
  but we have space to take it back in.}

\section{Preliminaries}
\setCoqFilename{Prelims}

Everything in this paper is carried out
in Coq's type theory and all reasoning is constructive. \todo{explain $\Prop$ and $\bot$}
We use the following inductive types:
$\nat$ providing the \emph{numbers}
  $n::=0\mid\natS n$, and
 $\opt(X)$ providing the \emph{options} 
  $\none$ and $\some x$\todo{mention that missing cases default to \none}, and
 $\List(X)$ providing the \emph{lists}
 $A::=\nil\mid x::A$.
 
For lists $A,B:\List(X)$ we use the functions 
\emph{length} $|A|:\nat$,
\emph{concatenation} $A\con B:\List(X)$,
\emph{map} $f@A:\List(Y)$ where $f:X\to Y$, and
\emph{lookup} $A[n]:\opt(X)$ where 
${(x::A)[0]=\some x}$,
and ${(x::A)[\natS n]=A[n]}$,
and $\nil[n]=\none$.
When we define functions that yield an option,
we will omit equations that yield $\none$
(e.g., the third equation $\nil[n]=\none$ 
defining lookup $A[n]:\opt(X)$ will be omitted).

We write $\Prop$ for the universe of propositions
and $\bot$ for the proposition falsity.
A \emph{relation on $X$ and $Y$} is 
a predicate $X\to Y\to\Prop$,
and a \emph{relation on $X$} is
a predicate $X\to X\to\Prop$.
A relation $R$ is \coqemph[functional]{functional}
if $y=y'$ whenever $Rxy$ and $Rxy'$.
A relation $R$ on $X$ and~$Y$ is \coqemph[computable]{computable}
if there is a function $f:X\to\opt(Y)$ such that 
$  \forall x.~ 
  (\exists y.~fx=\some y\land Rxy)\lor
  (fx=\none\land\neg\exists y.~Rxy)
$.

We use a recursive \emph{membership} predicate $x\in A$
such that $(x\in\nil)=\bot$ and 
$(x\in y::A)=(x{=}y\lor x\in A)$.


We define an inductive predicate \emph{$\M{ter}_R~x$}
identifying the \coqemph[terminatesOn]{terminating points}
of a relation $R$ on~$X$:
\begin{mathpar}
  \inferrule*
  {\forall x'.~Rxx'\to\M{ter}_R~x'}
  {\M{ter}_R~x}
\end{mathpar}
If $x$ is a terminating point of $R$,
we say that \coqemph[terminatesOn]{$R$ terminates on $x$}
or that \coqemph[terminatesOn]{$x$ terminates for $R$}.
We call a relation \emph{terminating}
if it terminates on every point.

Let $R$ be a relation on $X$.
The \coqemph[evaluates]{span of $R$} is 
the inductive relation \coqemph[evaluates]{$\eval_R$} on $X$
defined as follows:
\begin{mathpar}
 \inferrule*{\neg\exists y.~Rxy}{x\eval_R x}
 \and
 \inferrule*{Rxx'\\x'\eval_R y}{x\eval_R y}
\end{mathpar}
If $x\eval_R y$,
we say that \coqemph[evaluates]{$y$ is a normal form of $x$ for $R$}.

\begin{fact}[][evaluates_fun]\label{fact-eval}
  \begin{enumerate}
  \coqitem[evaluates_fun] If $R$ is functional, then $\eval_R$ is functional.
  \coqitem[normalizes_terminates] 
  If $R$ is functional and $x$ has a normal form for $R$,
    then $R$ terminates on $x$.
  \coqitem[terminates_normalizes] If $R$ is computable, 
    then every terminating point of $R$
    has a normal form for~$R$.
  \end{enumerate}
\end{fact}
\makeproof{fact-eval}{
  (1) follows by induction on $\eval_R$.
  (2) follows by induction on $x\eval_Ry$.
  (3) follows by induction on $\M{ter}_R\,x$.
\qed}

A \coqemph[ARS]{reduction system} is a structure
consisting of a type $X$ and a relation $R$ on~$X$.
Given a reduction system $A=(X,R)$,
we shall write $A$ for the type $X$ and
$\red_A$ for the relation of $A$.
We say that \emph{$a$ reduces to $b$ in $A$}
if $a \red_Ab$.

\section{Call-by-value Lambda Calculus L}
\setCoqFilename{L}

The call-by-value $\lambda$-calculus
we consider in this paper
employs de Bruijn terms with
simple substitution and
admits only abstractions as values.

We provide \coqemph[term]{terms} with an inductive type
\begin{align*}
  s,t,u,v
  &~:~\Ter~::=~
    n\mid st\mid\lambda s
  \qquad(n:\nat)
\end{align*}
and define a recursive function \coqemph[subst]{$\subst sku$}
providing \emph{simple substitution}:
\begin{align*}
  \subst kku
  &~:=~u
  &&&&&\subst{(st)}ku
  &~:=~(\subst sku)(\subst tku)&&
  \\
  \subst nku
  &~:=~n
  &&\text{if}~n\neq k&&
  &
  \subst{(\lambda s)}ku
  &~:=~\lambda(\subst s{\natS k}u)&&
\end{align*}
We define an inductive 
\coqemph[stepL]{reduction relation $s\red t$}
on terms:
\begin{mathpar}
  \inferrule*{~}{(\lambda s) (\lambda t)\red\subst s0{\lambda t}}
  \and
  \inferrule*{s\red s'} {st\red s't}
  \and
  \inferrule* {t\red t'} {(\lambda s)t\red(\lambda s)t'}
\end{mathpar}

\begin{fact}[][stepL_funct]
  $s\red t$ is functional and computable.
\end{fact}

We define an 
\coqemph[boundL]{inductive bound predicate $s\bnd k$}
for terms:
\begin{mathpar}
  \inferrule*{n<k}{n\bnd k}
  \and
  \inferrule*{s\bnd k\\t\bnd k}{st\bnd k}
  \and
  \inferrule*{s\bnd\natS k}{\lambda s\bnd k}
\end{mathpar}
Informally, $s\bnd k$ holds if
every free variable of $s$ is smaller than $k$.
A term is \coqemph[closedL]{closed} if $s\bnd 0$.
A term is \emph{open} if it is not closed.

For closed terms, 
reduction in L agrees with
reduction in the $\lambda$-calculus.
For open terms,
reduction in L is ill-behaved
since L is defined with simple substitution. 
For instance, we have
$(\lambda\lambda 1)(\lambda 1)(\lambda 0)
\red(\lambda\lambda 1)(\lambda 0)
\red\lambda\lambda 0$.
Note that the second $1$ in the initial term is not bound and refers to the De Bruijn index $0$. Thus the first reduction step is capturing.

We define \coqemph[stuck]{stuck terms} inductively:
\begin{mathpar}
  \inferrule*~{\M{stuck}~n}
  \and
  \inferrule*{\M{stuck}~s}{\M{stuck}\,(st)}
  \and
  \inferrule*{\M{stuck}~t}{\M{stuck}\,((\lambda s)t)}
\end{mathpar}

\begin{fact}[Trichotomy][L_trichotomy]\label{fact-term-tricho}
  For every term $s$,
  exactly one of the following holds:\\
  (1)~$s$~is reducible. 
  (2)~$s$~is an abstraction.
  (3)~$s$~is stuck.
\end{fact}
\makeproof{fact-term-tricho}{
  By induction on $s$.
\qed}




\section{Machines and Refinements}
\setCoqFilename{Refinements}

We model machines as reduction systems.
Recall that L is also a reduction system.
We relate a machine M with L with a relation 
$a\gg s$ we call \textit{refinement}.
If $a\gg s$ holds, 
we say that $a$ (a state of $M$) 
refines $s$ (a term of L).
Correctness means that L can simulate steps of M 
such that refinement between states and terms is preserved.
Concretely, 
if $a$ refines $s$ and $a$ reduces to $a'$ in M,
then either $a'$ still \textit{refines} $s$
or~$s$ reduces to some $s'$ in L
such that $a'$ refines $s'$.
Steps where the refined term stays unchanged
are called \textit{silent}.

The general idea is now as follows.
Given a term $s$, 
we compile $s$ into a refining state~$a$.
We then run the machine on $a$.
If the machine terminates 
with a normal form~$b$ of~$a$,
we decompile $b$ into a term $t$
such that $b$ refines $t$ and
conclude that~$t$ is a normal form of $s$.
We require that the machine terminates
for every state refining a term 
that has a normal form.

\begin{definition}[][Machine]
  A \coqemph[Machine]{machine} is a structure 
  consisting of a type $A$ of \coqemph[M_A]{states} 
  and two relations \coqemph[M_rel]{$\tred$} and \coqemph[M_rel]{$\bred$}
  on~$A$. When convenient, we consider a machine $A$
  as a reduction system with the relation
  $\coqlink[M_rel]{\N{\red_A}}:={\tred}\cup{\bred}$.
\end{definition}

The letter $X$ ranges over reduction systems
and $A$ and $B$ range over machines.

\begin{definition}[][refinement_ARS]\label{definition-ref}
  A \coqemph[refinement_ARS]{refinement $A$~to~$X$} 
  is a relation $\gg$ on~$A$ and $X$ 
  such   that:
  \begin{enumerate}
  \item If $a\gg x$ and $x$ is reducible with $\red_X$, 
    then $a$ is reducible with $\red_A$.
  \item If $a\gg x$ and $a\tred a'$, then $a'\gg x$.
  \item If $a\gg x$ and $a\bred a'$, then there exists $x'$ such that
    $a'\gg x'$ and $x\red x'$.
  \item If $a\gg x$, then $a$ terminates for $\tred$.
  \end{enumerate}
  We say that \emph{$a$ refines $x$} if $a\gg x$.
\end{definition}

Figure~\ref{fig:simulation} 
illustrates refinements with a diagram.  
Transitions in $X$ appear in the upper line
and transitions in $A$ appear in the lower line.
The dotted lines represent the refinement relation.
Note that conditions~(2) and~(3) of
Definition~\ref{definition-ref} ensure
that refinements are bottom~up simulations
(i.e., $X$ can simulate~$A$).
Conditions~(1) and~(4) are progress conditions.
They suffice to ensure that refinements also
act as top-down simulations (i.e. $A$ can simulate $X$),
given mild assumptions 
that are fulfilled by L and all our machines.

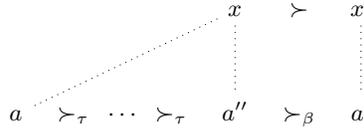
\begin{figure}[t]
  \[
    \begin{tikzcd}[column sep=tiny]
      &&x \ar[d, dash, dotted] 
      \ar[dll, dash, dotted] &\red &x' \ar[d, dash, dotted] \\[3mm]
      a&\tred~\cdots~\tred &a'' &\bred &a'
    \end{tikzcd}
  \]
  \vspace{-5mm}
  \caption{Refinement diagram}
  \label{fig:simulation}
\end{figure}

\begin{fact}[Correctness][upSim]
  Let $\gg$ be a refinement $A$ to $X$ and $a\gg x$.
  Then:
  \begin{enumerate}
  \coqitem[upSim] If $a\eval_A a'$, 
    there exists $x'$ such that
    $a'\gg x'$ and $x\eval_X x'$.
  \coqitem[rightValue] If $a\eval_A a'$,  $a'\gg x'$, 
    and $\gg$ is functional,
    then $x\eval_X x'$.
  \coqitem[termination_propagates] If $x$ terminates for $\red_X$,
    then $a$ terminates for $\red_A$.
  \coqitem[evaluation_propagates] If $x$ terminates for $\red_X$
    and $\red_A$ is computable, then
    there exists $a'$ such that
    $a\eval_A a'$.
  \end{enumerate}
\end{fact}
\begin{proof}
  (1) follows by induction on $a\eval_Ax$.
  (2) follows with (1) and Fact~\ref{fact-eval}.
  (3) follows by induction 
  on the termination of $x$ for $\red_X$
  and the termination of $a$ for $\tred$.
  (4) follows by induction 
  on the termination of $x$.
\qed\end{proof}

We remark that the concrete reduction systems
we will consider in this paper
are all functional and computable.
Moreover, all concrete refinements will be functional
and, except for the heap machine,
also be computable.

A refinement may be seen as the combination of
an invariant and a decompilation function.
We speak of an invariant since the fact 
that a state is a refinement of a term
is preserved by the reduction steps of the machine.

Under mild assumptions fulfilled in our setting,
the inverse of a refinement
is a stuttering bisimulation~\cite{Baier2008}.
The following fact asserts 
the necessary top-down simulation.

\begin{fact}[][one_downSim]
  Let $\gg$ be a refinement $A$ to $X$ where 
  $\red_X$ is functional and $\tred$ is computable.
  \begin{enumerate}
  \coqitem[one_downSim] If $a\gg x\red_Xx'$, 
    then there exist $a'$ and $a''$
    such that $a\eval_\tau a''\bred a'\gg x'$.
  \coqitem[downSim] If $a\gg x\eval_Xx'$, 
    then there exists $a'$ 
    such that $a\eval_A a'\gg x'$.
  \end{enumerate}
\end{fact}  
\begin{proof}
  (1) follows with Fact~\ref{fact-eval}.  
  (2) follows by induction on $x\eval_Xx'$ using~(1).
\qed\end{proof}

We will also refine machines with machines and
rely on a composition theorem that combines
two refinements $A$~to~$B$ and $B$~to~$X$
to a refinement $A$~to~$X$.
We define refinement of machines 
with strict simulation.

\begin{definition}[][refinement_M]
  A \coqemph[refinement_M]{refinement $A$ to $B$} is a
  relation $\gg$ on~$A$ and $B$ such that:
  \begin{enumerate}
  \item If $a\gg b$ and $b$ is reducible with $\red_B$, then
    $a$ is reducible with $\red_A$.
  \item If $a\gg b$ and $a\tred a'$, then there exists $b'$ such that
    $a'\gg b'$ and $b\tred b'$.
  \item If $a\gg b$ and $a\bred a'$, then there exists $b'$ such that
    $a'\gg b'$ and $b\bred b'$.
  \end{enumerate}
\end{definition}

\begin{fact}[Composition][composition]
  \label{fact-ref-comp}
  Let $\gg_1$ be a refinement $A$ to $B$
  and $\gg_2$ be a refinement $B$ to $X$.
  Then the composition 
  $\lam{ac}{\,\exists b.\,a \gg_1 b\land b \gg_2 c$}
  is a refinement $A$ to $X$.
\end{fact}

\section{Programs}
\setCoqFilename{Programs}

The machines we will consider execute programs.
Programs may be seen as lists of commands
to be executed one after the other.
Every term can be compiled into a program,
and programs that are images of terms 
can be decompiled.
There are commands for variables, 
abstractions, and applications.
We represent \coqemph[Pro]{programs} with 
a tree-recursive  inductive type so that
the command for abstractions can nest programs:
\begin{align*}
  {\coqlink[Pro]{\N{P},\N{Q},\N{R}~:~\N{\Pro}}}~::=~
  \ret\mid
  \var\,n;P\mid
  \lamb\,Q;P\mid
  \app;P
  \qquad(n:\nat)
\end{align*}

We define a tail recursive
\coqlink[gamma]{\emph{compilation function} $\gamma:\Ter\to\Pro\to\Pro$}
translating terms into programs:
\begin{align*}
  \gamma nP~&:=~\var\,n;P & \gamma(\lambda s)P~&:=~\lamb(\gamma s\ret);P&
  \\
  \gamma(st)P~&:=~\gamma s(\gamma t(\app;P))  &&&&
\end{align*}
The second argument of $\gamma$ 
may be understood as a continuation.

We also define a \coqlink[delta]{\emph{decompilation function} $\delta
  P A$ of type $\Pro\to\List(\Ter)\to\opt(\List(\Ter))$}
translating programs into terms.
The function executes the program 
over a stack of terms.
The optional result acknowledges the fact
that not every program represents a term.  
We write~\emph{$A$} and~\emph{$B$} for lists of terms.
Here are the equations defining the decompilation function:
\begin{align*}
  \delta\,\ret\, A&~:=~\some A\\
  \delta (\var\,n;P)A&~:=~\delta P(n::A)\\
  \delta (\lamb\,Q;P) A&~:=~\delta P (\lambda s::A)
  &&\hspace{-2em}\text{if }\delta\,Q\,\nil=\some{[s]}\\
  \delta (\app;P) A&~:=~\delta P (st::A')
  &&\hspace{-2em}\text{if }A=t::s::A'
\end{align*}
Decompilation inverts compilation:

\begin{fact}[][decompile_correct']\label{fact-decompile}
  $\delta(\gamma sP) A=\delta P(s::A)$.
\end{fact}
\makeproof{fact-decompile}{
  By induction on $s$.
}

\begin{fact}[][decompile_append]\label{fact-delta-progr-ext-stack}
  Let $\delta PA=\some{A'}$.
  Then $\delta P(A\con A'')=\some{(A'\con A'')}$.
\end{fact}
\makeproof{fact-delta-progr-ext-stack}{
  By induction on $P$.
\qed}

We define a predicate $\N{P\gg s}:=~\delta P\nil=\some{[s]}$ 
read as \emph{$P$ represents $s$}.


The naive stack machine will use a
\coqlink[substP]{\emph{substitution operation $\subst PkR$}}
for programs:
\begin{align*}
  \subst{\ret\,}kR
  &~:=~\ret
  &&&  \subst{(\lamb\,Q;P)}kR
  &~:=~\lamb(\subst Q{\natS k}R);\subst PkR&&
  \\
  \subst{(\var\,k;P)}kR
  &~:=~\lamb\,R;\subst PkR
  &&&\subst{(\app;P)}kR
  &~:=~\app;\subst PkR &&
  \\
  \subst{(\var\,n;P)}kR
  &~:=~\var\,n;\subst PkR
  &&\text{if}~n\neq k  &&
\end{align*}
Note the second equation for the variable command
that replaces a variable command with a lambda command.
The important thing to remember here is 
the fact that the program $R$ 
is inserted as the body of a lambda command.

For the verification of the naive stack machine
we need a substitution lemma relating 
term substitution with program substitution.
The lemma we need appears as 
\Cref{corollary-program-subst} below.
We prove the fact with a generalised version
that can be shown by induction on programs.
We use the notation $\N{\subst Aku}:=(\lam s{\subst sku})@A$.

\begin{lemma}[Substitution][substP_rep_subst']
  \label{lem-beta-subst1}
  Let $R\gg t$ and $\delta QA=\some B$.
  Then $\delta\,\subst QkR\,\subst Ak{\lambda t}
  =\some{\subst Bk{\lambda t}}$.
 \end{lemma}
\makeproof{lem-beta-subst1}{
  By induction on $Q$.
  We show the case for $Q=\lamb\,P;Q$.

  Let $\delta(\lamb\,P;Q)A=\some B$.
  Then $P\gg s$ and $\delta Q(\lambda s::A)=\some B$ for some $s$.
  Now:
  \begin{align*}
  \delta\,\subst{(\lamb\,P;Q)}kR\,
  \subst Ak{\lambda t}
    &~=~
      \delta\,(\lamb\,\subst P{\natS k}R;\subst QkR)\,
      \subst Ak{\lambda t}
    \\&~=~
      \delta\,\subst QkR\,
      (\lambda\subst s{\natS k}{\lambda t}::\subst Ak{\lambda t})
    &&\text{since $\subst P{\natS k}R\gg\subst s{\natS k}{\lambda t}$ 
       by ind. hyp. for $P$}
    \\&~=~
      \delta\,\subst QkR\,
      \subst {(\lambda s::A)}k{\lambda t}
    \\&~=~\some{\subst Bk{\lambda t}}
    &&\text{ind. hyp. for $Q$}
  \end{align*}
  \qed
}

\begin{corollary}[Substitution]
  \label{corollary-program-subst}
  If $P\gg s$ and $Q\gg t$, 
then $\subst PkQ\gg\subst sk{\lambda t}$.
\end{corollary}

We define a \coqlink[boundP]{\emph{bound predicate $P\bnd k$}}
for programs 
that is analogous to the bound predicate for terms and say that a program $P$ is \coqemph[closedP]{closed} if $P\bnd0$:
\begin{mathpar}
  \inferrule*{~}{\ret\bnd k}
  \and
  \inferrule*{n<k\\P\bnd k}{\var\,n;P\bnd k}
  \and
  \inferrule*
  {Q \bnd\natS k \\P\bnd k}
  {\lamb\,Q;P\bnd k}
  \and
  \inferrule*{P\bnd k}{\app;P\bnd k}
\end{mathpar}

\begin{fact}[][bound_compile]\label{fact-gamma-bnd}
  If $s\bnd k$ and $P\bnd k$,
  then $\gamma sP\bnd k$.
\end{fact}
\makeproof{fact-gamma-bnd}{
  By induction on $s$.
}
It follows that $\gamma sP$ is closed 
whenever $s$ and $P$ are closed.

\section{Naive Stack Machine}
\label{sec:naive-stack-machine}
\setCoqFilename{M_stack}

The naive stack machine executes programs
using two stacks of programs 
called \emph{control stack} 
and \emph{argument stack}.
The control stack holds the programs to be executed,
and the argument stack holds the programs 
computed so far.
The machine executes the first command
of the first program on the control stack
until the control stack is empty 
or execution of a command fails.

The \coqlink[stateS]{\emph{states of the naive stack machine}} are pairs
\begin{align*}
  \N{(T,V)}&~:~\List(\Pro)\times\List(\Pro)
\end{align*}
consisting of two lists $T$ and $V$ representing
the control stack and the argument stack.
We use the letters $T$ and $V$ since 
we think of the items on $T$ as tasks
and the items on $V$ as values.
The \coqlink[stepS]{\emph{reduction rules of the naive stack machine}}
appear in Figure~\ref{fig:stack-red}.
The parentheses for states 
are omitted for readability.
We will refer to the rules as
\emph{return rule},
\emph{lambda rule}, and
\emph{application rule}.
The return rule 
removes the trivial program from the control stack.
The lambda rule
pushes a program representing an abstraction 
on the argument stack.
Note that the programs on the control stack 
are executed as they are.
This is contrast to the programs on
the argument stack that represent
bodies of abstractions.
The application rule
takes two programs from the argument stack
and pushes an instantiated program obtained by $\beta$-reduction 
on the control stack.
This way control is passed from
the calling program to the called program.
There is no reduction rule for the variable command \todo{explain why}
since we will only consider states that represent closed terms.

\begin{figure}[t] 
  \begin{align*}
    \ret::T,~V
    &~\tred~T,~V
    \\
    (\lamb\,Q;P)::T,~V
    &~\tred~P::T,~Q::V
    \\
    (\app;P)::T,~R::Q::V
    &~\bred~\subst Q0R::P::T,~V
  \end{align*}
  \vspace{-7mm}
  \caption{Reduction rules of the naive stack machine}
  \label{fig:stack-red}
\end{figure}

\begin{fact}[][tau_functional]\label{fact-stack-beta-tau-fun-term}
  The relations $\tred$, $\bred$, 
  and ${\tred}\cup {\bred}$
  are functional and computable.  
  Moreover, the relations $\tred$ and $\bred$
  are terminating.
\end{fact}

We decompile machine states by executing
the task stack on the stack of terms 
obtained by decompiling the programs
on the value stack.
To this purpose we define two decompilation functions.
The \coqlink[deltaV]
{\emph{decompilation function $\delta V$ for argument stacks}} has type $\List(\Pro)\to\opt(\List(\Ter))$
and satisfies the equations
\begin{align*}
  \delta\nil
  &~:=~\some\nil
  \\
  \delta(P::V)
  &~:=~\some{(\lambda s::A)}
  &&\text{if}~P\gg s
     ~\text{and}~
     \delta V=\some{A}
\end{align*}
Note that the second equation turns 
the term~$s$ obtained from a program on the argument stack
into the abstraction $\lambda s$.
This accounts for the fact that programs
on the argument stack represent bodies of abstractions.
The \coqlink[deltaT]
{\emph{decompilation function  $\delta T A$ for control stacks}} 
has type $\List(\Pro) \to\List(\Ter)\to\opt(\List(\Ter))$
and satisfies the equations
\begin{align*}
  \delta\nil A
  &~:=~\some A
  \\
  \delta(P::T)A
  &~:=~\delta TA'
  &&\text{if}~\delta PA=\some{A'}
\end{align*}

We now define the \emph{refinement relation}
between states of the naive stack machine
and terms as follows:
\begin{align*}
  \coqlink[repsSL]{\N{(T,V)\gg s}}
  &~:=~\exists A.~\delta V=\some A~\land~\delta TA=\some{[s]}
\end{align*}
We will show that $(T,V)\gg s$ is in fact a refinement.

\begin{fact}[][repsSL_functional]
  \label{fact-stack-ref-fun-comp}
  $(T,V)\gg s$ is functional and computable.
\end{fact}

\begin{fact}[$\tau$-Simulation][tau_simulation] \label{fact-stack-tau-sim}
  If $(T,V)\gg s$ and $T,V\tred T',V'$, 
    then $(T',V')\gg s$.
\end{fact}
\begin{proof}
  We prove the claim for the second $\tau$-rule,
  the proof for the first $\tau$-rule is similar.
  Let $\lamb\,Q;P::T,\,V~\tred~P::T,\,Q::V$.
  We have
  \begin{align*}
    \delta(\lamb\,Q;P::T)(\delta V)~=~
    &\delta T(\delta(\lamb\,Q;P)(\delta V))
    \\
    ~=~
    &\delta T (\delta P (\lambda s::\delta V))
    &&Q\gg s 
   \\
    ~=~
    &\delta (P::T)(\delta (Q::V))\tag*{\qed}
  \end{align*}
\end{proof}

Note that the equational part of the proof
nests optional results
to avoid cluttering with side conditions and 
auxiliary names.

Proving that L can simulate $\beta$-steps
of the naive stack machine takes effort.

\begin{fact}[][decompileArg_abstractions]\label{fact-delta-V-abs}
  If $\delta V=\some A$,
  then every term in $A$ is an abstraction.
\end{fact}

\begin{fact}\label{lem-beta-app}
  $\delta(\app;P::T)(t::s::A)=\delta(P::T)(st::A)$.
\end{fact}
\makeproof{lem-beta-app}{
  $\delta(\app;P::T)(t::s::A)=
  \delta T(\delta(\app;P)(t::s::A))=
  \delta T(\delta P(st::A))=
  {\delta(P::T)(st:A)}$.
}

\begin{fact}\label{fact-delta-prog-ext}
  $\delta(P::T)A=\delta T(s::A)$ if $P\gg s$.
\end{fact}
\begin{proof}
  Follows with Fact~\ref{fact-delta-progr-ext-stack}.
\qed\end{proof}

\begin{lemma}[Substitution][substP_rep_subst]
  \label{lem-beta-subst}
  $\delta(\subst Q0R::T)A=\delta T(\subst s0{\lambda t}::A)$
  if $Q\gg s$ and $R\gg t$. 
\end{lemma}
\begin{proof}
  By \Cref{corollary-program-subst}
  we have $\subst Q0R\gg\subst s0{\lambda t}$.
  The claim follows with Fact~\ref{fact-delta-prog-ext}.
\qed\end{proof}

We also need a special 
reduction relation \coqemph[stepLs]{$A\red A'$} 
for term lists:
\begin{mathpar}
  \inferrule*
  {s\red s'\\\forall t\in A.~t~\text{is an abstraction}}
  {s::A\red s'::A}
  \and
  \inferrule*
  {A\red A'}
  {s::A\red s::A'}
\end{mathpar}
Informally,
$A\red A'$ holds if $A'$ can be obtained from $A$
by reducing the term in $A$ 
that is only followed by abstractions.

\begin{lemma}[][stepLs_decomp]\label{lem-beta-red}
  Let $A\red A'$ and $\delta PA=\some B$.
  Then $\exists B'.~B\red B'\land\delta PA'=\some{B'}$.
\end{lemma}
\begin{proof}
  By induction on $P$.
  We consider the case $P=\app;P$.

  Let $\delta(\app;P)(t::s::A)=\some B$
  and $t::s::A\red t'::s'::A'$.
  Then $\delta P(st::A)=\some B$
  and $st::A\red s't'::A'$
  (there are three cases:
  (1) $t=t'$, $s=s'$, and $A\red A'$; 
  (2)~$t=t'$, $s>s'$,  $A=A'$, and
  $A$ contains only abstractions; 
  (3) $t\red t'$, ${s::A=s'::A'}$, and
  $s::A$ contains only abstractions).
  By the inductive hypothesis we have
  $B\red B'$ and ${\delta P(s't'::A')=\some{B'}}$
  for some $B'$.
  Thus $\delta(\app;P)(t'::s'::A')=\some{B'}$.
\qed\end{proof}

\begin{fact}[$\beta$-Simulation][beta_simulation]
  \label{fact-beta-corr}~\\
  If $(T,V)\gg s$ and $T,V\bred T',V'$, 
  then ${\exists s'.~(T',V')\gg s'\land s\red s'}$.\todo{proofread partial replacement of s with u by second pair of eyes} \todo{Proof: why exists final s'?}
\end{fact}
\begin{proof}
  Let $\app;P::T,\,R::Q::V~\bred~\subst Q0R::P::T,\,V$. 
  Moreover, let $R\gg t$, \ $Q\gg u$, and $\delta V=\some A$.
  We have:
  \begin{align*}
    \some{[s]}~=~
    &\delta(\app;P::T)(\delta(R::Q::V))
    \\~=~
    &\delta(\app;P::T)(\lambda t::\lambda u::A)
    \\~=~
    &\delta(P::T)((\lambda u)(\lambda t)::A)
    &&\text{Fact~\ref{lem-beta-app}}
    \\~\red~
    &\delta(P::T)(\subst u0{\lambda t}::A)
    &&\text{Lemma~\ref{lem-beta-red} and Fact~\ref{fact-delta-V-abs}}
   \\~=~
    &\delta(\subst Q0R::P::T)A
    &&\text{Lemma~\ref{lem-beta-subst}}
    \\~=~
    &\some{[s']}
    &&\text{for some $s'$}
  \end{align*}
  Note that $s'$ exists since $\red$ preserves
  the length of a list.
  We now have $s\red s'$ by the definition of $\red$
  and $(\subst Q0R::P::T,\,V)\gg s'$,
  which concludes the proof.
\qed\end{proof}

It remains to show that states are reducible
if they refine reducible terms.
For this purpose, we define \coqemph[stuckLs]{stuck term lists}:
\begin{mathpar}
  \inferrule*
  {\M{stuck}~s\\\forall t\in A.~t~\text{is an abstraction}}
  {\M{stuck}\,(s::A)}
  \and
  \inferrule*
  {\M{stuck}~A}
  {\M{stuck}\,(s::A)}
\end{mathpar}
Note that $s$ is stuck iff $[s]$ is stuck.

\begin{lemma}[][stuck_decompile]\label{lemma-stuck-program}
  Let $A$ be stuck and $\delta PA=\some B$.
  Then $B$ is stuck.  
\end{lemma}
\makeproof{lemma-stuck-program}{
  By induction on $P$.
}

\begin{lemma}[][stuck_decompileTask]\label{fact-stuck-task-stack}
  Let $A$ be stuck and $\delta TA=\some B$.
  Then $B$ is stuck.  
\end{lemma}
\makeproof{fact-stuck-task-stack}{
  By induction on $T$ using 
  Lemma~\ref{lemma-stuck-program}.
\qed}

\begin{fact}[Trichotomy][stateS_trichotomy]\label{fact-stack-tricho}
  Let $T,V \gg s$. Then exactly one of the following holds:
  \begin{enumerate}
  \item $(T,V)$ is reducible.
  \item $(T,V)=(\nil,[P])$ and $P \gg s'$ with $s=\lambda s'$ for some $P,s'$.
  \item $T=\var\,x;P::T'$ for some $x,P,T'$ and $s$ is stuck.
  \end{enumerate}
\end{fact}  
\begin{proof}
  Let $\delta V=\some A$ and
  $\delta TA=\some{[s]}$, and $s$ be reducible.
  By Fact~\ref{fact-delta-V-abs} we know
  that $A$ contains only abstractions.
  Case analysis on $T$.

  $T=\nil$.  Then $A=[s]$ and the second case holds by definition of $\delta$.

  $T=\ret::T'$.  Then $(T,V)$ is reducible.

  $T=\M{var}\,n;P::T'$.
  We have
  $$
  \some{[s]}=
  \delta(\M{var}\,n;P::T')A=
  \delta T'(\delta(\M{var}\,n;P)A)=
  \delta T'(\delta P(n::A))
  $$
  Since $n::A$ is stuck,
  we know by Lemmas~\ref{lemma-stuck-program}
  and~\ref{fact-stuck-task-stack}
  that $[s]$ is stuck. Thus the third case holds.

  $T=\lamb\,Q;P::T'$.  Then $(T,V)$ is reducible.

  ${T=\app;P::T'}$.
  Then $\some{[s]}=\delta(\app;P::T')A=\delta T'(\delta(\app;P)A)$
  and hence $A=t::s::A'$.  Thus $V=R::Q::V'$.
  Thus $(T,V)$ is reducible.
\qed\end{proof}

\begin{corollary}[Progress][reducible_red]
  \label{fact-stack-reducibility}
  If $T,V\gg s$ and $s$ is reducible,
  then $(T,V)$ is reducible.
\end{corollary}
\begin{proof}
  Follows from Fact~\ref{fact-stack-tricho} using Fact~\ref{fact-term-tricho}.
\qed\end{proof}

\begin{theorem}[Naive Stack Machine to L][stack_L_refinement]
  \label{theo-stack-L}
  The relation
  \begin{align*}
    (T,V)\gg s
    &~:=~\exists A.~\delta V=\some A~\land~\delta TA=\some{[s]}
  \end{align*}
  is a functional and computable refinement.
  Moreover, \coqlink[compile_stack_L]{$([\gamma\,s\,\ret],\nil)\gg s$}
  holds for every term $s$.
\end{theorem}
\begin{proof}
  The first claim follows with 
  Facts~\ref{fact-stack-ref-fun-comp},
  \ref{fact-stack-reducibility},
  \ref{fact-stack-tau-sim},
  \ref{fact-beta-corr},
  and~\ref{fact-stack-beta-tau-fun-term}.
  The second claim follows with Fact~\ref{fact-decompile}.
\qed\end{proof}

\section{Closures}
\setCoqFilename{Closures}

A closure is a pair consisting of 
a program and an environment.
An environment is a list of closures
representing a delayed substitution.
With closures we can refine 
the naive stack machine so that no
substitution operation is needed.
\begin{align*}
  \N{e}~
  &:~\coqlink[Clo]{\N{\Clo}}~::=~P/E
  &&\text{\coqemph[Clo]{closure}}
  \\
  \N{E},\N{F},\N{T},\N{V}~
  &:~\List{(\Clo)}
  &&\text{\emph{environment}}
\end{align*}

For the decompilation of closures 
into plain programs we define a 
\coqlink[substPl]
{\emph{parallel substitution operation $\subst PkW$}}
for programs
($W$ ranges over lists of programs):
\begin{align*}
  \subst{\ret\,}kW
  &~:=~\ret \\
  \subst{(\app;P)}kW
  &~:=~\app;\subst PkW
  \\
  \subst{(\lamb\,Q;P)}kW
  &~:=~\lamb(\subst Q{\natS k}W);\subst PkW
  \\
  \subst{(\var\,n;P)}kW
  &~:=\text{if}~~n\geq k \land W[n-k]=\some Q \text{
    then}~\lamb\,Q;\subst PkW \text{ else}~\var\,n;\subst PkW
\end{align*}

We will use the notation
$\N{W\!\bnd 1}:=\forall P\in W.~P\bnd 1$.

\begin{fact}[Parallel Substitution][substPl_nil]\label{fact-par-subst}
  \begin{enumerate}
  \coqitem[substPl_nil] $\subst Pk\nil=P$.
  \coqitem[boundP_mono] If $P\bnd k$ and $k\le k'$, then $P\bnd k'$.
  \coqitem[boundP_substP] If $P\bnd k$, then $\subst PkQ=P$.
  \coqitem[substPl_cons] If $W\!\bnd 1$, then 
    $\subst Pk{Q::W}=\subst{(\subst P{\natS k}W)}kQ$.
  \coqitem[substP_boundP] If $W\!\bnd 1$ and $P\bnd |W|+k$, 
    then $\subst PkW\bnd k$.
  \end{enumerate}
\end{fact}
\makeproof{fact-par-subst}{
  (1) follows by induction on $P$. 
  (2) and (3) follow by induction on $P\bnd k$.
  (4) follows by induction on $P$ 
  using (2) and (3) in the variable case.
  (5) follows by induction on $P$.
}

Note that Fact~\ref{fact-par-subst}\,(4)
relates parallel substitution to single substitution
.

We define a function 
$\coqlink[deltaC]{\N{\delta_1 e}}$ of type $\Clo\to\Pro$
translating closures into programs:
\begin{align*}
  \coqlink[deltaC]{\N{\delta_1 (P/E)}}
  &~:=~\subst P 1 {\delta_1@E}
\end{align*}

We also define an inductive
\coqemph[boundC]{bound predicate $e\bnd1$} 
for closures:
\begin{mathpar}
  \inferrule
  {P\bnd\natS |E|\\E\bnd1}
  {P/E\bnd1} \hspace{2em}
  {\N{E\bnd 1}~:=~\forall e\in E.~e\bnd 1}
\end{mathpar}
Note the recursion through environments
via the map function and via the membership predicate in the last two definitions.

\begin{fact}[][translateC_boundP]
  \label{fact-closure-subst}
  If $e\bnd1$, then $\delta_1 e\bnd1$. 
\end{fact}
\makeproof{fact-closure-subst}{
\begin{proof}
  Follows by induction on $e\bnd1$ 
  using Fact~\ref{fact-par-subst}\,(5).
 \qed\end{proof}}

\section{Closure Machine}
\label{sec:closure-machine}
\setCoqFilename{M_clos}

We now refine the naive stack machine by replacing 
all programs on the control stack and the argument stack
with closures, eliminating program substitution. 

\begin{figure}[t]\footnotesize
  \begin{align*}
    (\ret/E)::T,~V
    &~\tred~T,~V
    \\
    (\var\,n;P/E)::T,~V
    &~\tred~(P/E)::T,~e::V 
    &&\text{if}~E[n]=\some e
    \\
    (\lamb\,Q;P/E)::T,~V
    &~\tred~(P/E)::T,~(Q/E)::V
    \\
    (\app;P/E)::T,~e::(Q/F)::V
    &~\bred~(Q/e::F)::(P/E)::T,~V
  \end{align*}%
  \vspace{-7mm}
  \caption{Reduction rules of the closure machine}%
  \label{fig:closure-red}
\end{figure}

\coqemph[stateC]{States of the closure machine} are pairs 
\begin{align*}
  \N{(T,V)}&~:~\List(\Clo)\times\List(\Clo)
\end{align*}
consisting of a \emph{control stack~$T$} 
and an \emph{argument stack~$V$}.

The \coqemph[stepC]{reduction rules of the closure machine}
appear in Figure~\ref{fig:closure-red}.
The \emph{variable rule} (second $\tau$-rule) is new.
It applies if the environment provides a closure for the variable.
In this case the closure
is pushed on the argument stack.
We see this as delayed substitution of the variable.
The variable rule will be simulated
with the lambda rule of the naive stack machine.

The \emph{application rule} ($\beta$-rule)
takes two closures $e$ and $Q/F$ from the argument stack
and pushes the closure $Q/e::F$ on the control stack,
which represents the result of $\beta$-reducing
the abstraction represented by $Q/F$ 
with the argument $e$.

We will show that the closure machine 
implements the naive stack machine correctly
provided there are no free variables.

There is the complication that the closures on
the control stack must be closed while the
closures on the argument stack are allowed
to have the free variable $0$ representing 
the argument to be supplied by the application rule.

We define \coqemph[closedSC]{closed states} 
of the closure machine as follows:
\begin{align*}
  \N{P/E\bnd0}
  &~:=~P\bnd|E|~\land~E\bnd1
  \\
  \N{T\bnd 0}
  &~:=~\forall e\in T.~e\bnd 0
  \\
  \coqlink[closedSC]{\N{\M{closed}\,(T,V)}}
  &~:=~T\bnd0~\land~V\bnd1
\end{align*}
We define a function $\N{\delta_0 e}$ of type $\Clo\to\Pro$
for decompiling closures on the task stack:
\setCoqFilename{Closures}%
\begin{align*}
  \coqlink[deltaC]{\N{\delta_0 (P/E)}}
  &~:=~\subst P 0 {\delta_1@E}
\end{align*}
\setCoqFilename{M_clos}%

We can now define the \emph{refinement relation}
between states of the closure machine 
and states of the naive stack machine:
\begin{align*}
  \coqlink[repsCS]{\N{(T,V)\gg\sigma}}
  &~:=~\M{closed}\,(T,V)~\land~
    (\delta_0@T,\delta_1@V)=\sigma
\end{align*}
We show that $(T,V)\gg\sigma$ is a refinement. 

\begin{fact}[][repsCS_functional]\label{fact-closure-ref-fun-comp}
  $(T,V)\gg\sigma$ is functional and computable.
\end{fact}

\begin{fact}[Progress][reducibility]
  \label{fact-closure-red}
  Let 
  $(\delta_0@T,\delta_1@V)$ be reducible.
  Then $(T,V)$ is reducible.
\end{fact}
\makeproof{fact-closure-red}{
  Straightforward case analysis.
  \qed}

\begin{fact}[][closedSC_preserved]\label{fact-closure-closed}
  Let $(T,V)$ be closed and $(T,V)\red(T',V')$.
  Then $(T',V')$ is closed.
\end{fact}
\makeproof{fact-closure-closed}{
  Straightforward case analysis.
}

\begin{fact}[$\tau$-Simulation][tau_simulation]
  \label{fact-closure-tau}
  Let $(T,V)\tred(T',V')$.\\
  Then $(\delta_0@T,\delta_1@V) \tred(\delta_0@T',\delta_1@V')$.
\end{fact}
\makeproof{fact-closure-tau}{
  Straightforward case analysis.
}

\begin{fact}[$\beta$-Simulation][beta_simulation]
  \label{fact-closure-beta}
  Let $(T,V)$ be closed and $(T,V)\bred(T',V')$.\\
  Then $(\delta_0@T,\delta_1@V) \bred(\delta_0@T',\delta_1@V')$.
\end{fact}

\begin{proof}
  Follows with Facts~\ref{fact-par-subst}\,(4) 
  and~\ref{fact-closure-subst}.
\qed\end{proof}

\begin{theorem}[Closure Machine to Naive Stack Machine][clos_stack_refinement]
  \label{theo-closure-stack}
  The relation
  \begin{align*}
    (T,V)\gg\sigma
    &~:=~\M{closed}\,(T,V)~\land~
      (\delta_0@T,\delta_1@V)=\sigma
  \end{align*}
  is a functional and computable refinement.
  Moreover, 
  $([P/\nil],\nil)\gg([P],\nil)$
  holds for every closed program~$P$.
\end{theorem}
\begin{proof}
  The first claim follows with Facts
  \ref{fact-closure-ref-fun-comp}
  \ref{fact-closure-red},
  \ref{fact-closure-closed}, 
  \ref{fact-closure-tau}, 
  and~\ref{fact-closure-beta}.
  The second claim follows with 
  Fact~\ref{fact-par-subst}\,(1).
  \qed\end{proof}

Note that Theorems~\ref{theo-stack-L}
  and~\ref{theo-closure-stack}
  Facts~\ref{fact-ref-comp}
  and~\ref{fact-gamma-bnd} yield a refinement \coqlink[clos_L_refinement]{to L}.

  

\section{Codes}
\label{sec:codes}
\setCoqFilename{Codes}

If a state is reachable from an initial state in the closure machine,
all its programs are subprograms of programs in the initial state.
We can thus represent programs as addresses of a fixed code, providing
structure sharing for programs.

A code represents a program such that
the commands and subprograms of the program
can be accessed through addresses.
We represent codes abstractly
with a type $\Code$,
a type $\PA$ of program addresses,
and two functions $\#$ and $\phi$ as follows:
\begin{align*}
  \N{C}
  &~:~\N{\Code}
  &\text{\coqemph[Code]{code}}  
  \\
  \N{p},\N{q},\N{r}
  &~:~\N{\PA}
  &\text{\coqemph[PA]{program address}}
  \\
  \N{\#}
  &~:~\PA\to\PA
  \\
  &\hskip5mm\coqlink[Com]{\N{\Com}}~:=~\ret\mid\var\,n\mid\lamb\,p\mid\app
  &\text{\coqemph[Com]{command}}
  \\
  \N{\phi}
  &~:~\Code\to\PA\to\opt{(\Com)}
\end{align*}
Note that commands are obtained with
a nonrecursive inductive type $\Com$.
The function $\#$ increments a program address,
and the  $\phi$ yields the command 
for a valid program address.
We will use the notation $\N{C[p]}:=\phi Cp$.
We fix the semantics of codes with a relation \coqemph[representsPro]{$p\gg_CP$} relating program addresses with programs:
\begin{mathpar}
  \inferrule*
  {C[p]=\some{\ret}}
  {p\gg_C\ret}
  \and
  \inferrule*
  {C[p]=\some{\var\,n}\\\#p\gg_CP}
  {p\gg_C\var\,n;P}
  \\
  \inferrule*
  {C[p]=\some{\lamb\,q}\\q\gg_CQ\\\#p\gg_CP}
  {p\gg_C\lamb\,Q;P}
  \and
  \inferrule*
  {C[p]=\some{\app}\\\#p\gg_CP}
  {p\gg_C\app;P}
\end{mathpar}

\begin{fact}[][representsPro_functional]\label{fact-code-rep-fun}
  The relation $p\gg_C P$ is functional.
\end{fact}

We obtain one possible implementation of codes
as follows: \todo{intuition $n+k$}
\begin{align*}
  \PA
  &~:=~\nat
  &&&&&  \phi Cn
  &~:=~\lamb\,(n+k)
  &&\text{if}~C[n]=\lamb\,k
  \\
  \Code
  &~:=~\List(\Com)
  &&&&&  \phi Cn
  &~:=~C[n]
  &&\text{otherwise}
  \\
  \#n
  &~:=~\natS\,n
\end{align*}
For this realisation of codes we define
a function $\coqlink[psi]{\N{\psi}}:\Pro\to\List(\Com)$
compiling programs into codes as follows:
\begin{align*}
  \psi~\ret
  &~:=~[\ret]
  &\psi(\lamb\,Q;P)
  &~:=~\lamb\,(\natS|\psi P|)::\psi P\con\psi Q 
  \\
  \psi(\var\,n;P)
  &~:=~\var\,n::\psi P
  &\psi(\app;P)
  &~:=~\app::\psi P
\end{align*}
The linear representation of a program $\lamb\,Q;P$ 
provided by $\psi$ is as follows:
First comes a command $\lamb\,k$, then the commands for $P$, 
and finally the commands for $Q$ 
(i.e., the commands for the body $Q$ 
come after the commands for the continuation~$P$).  
The number $k$ of the command $\lamb\,k$ 
is chosen such that $n+\natS k$ 
is the address of the first command for $Q$
if $n$ is the address of the command $\lamb\,k$.

\begin{fact}[][fetch_correct]\label{fact-code-correct}
  $|C_1|\gg_{C_1\con\psi P\con C_2}P$. In particular, $0\gg_{\psi P}P$.
\end{fact}
\makeproof{fact-code-correct}{
  The generalisation 
  follows by induction on $P$.}

\section{Heaps}
\label{sec:heaps}
\setCoqFilename{Heaps}

A heap contains environments accessible through addresses.
This opens the possibility to share
the representation of environments.

We model heaps abstractly 
based on an assumed code structure.
We start with types for heaps 
and heap addresses and a function $\M{get}$ 
accessing heap addresses:
\begin{align*}
  \N{H}
  &~:~\N{\Heap}
  &\text{\coqemph[Heap]{heap}}
  \\
  \N{a},\N{b},\N{c}
  &~:~\N{\HA}
  &\text{\coqemph[HA]{heap address}}
  \\
  \N{g}
  &~:~\N{\HC}~:=~\PA\times\HA
  &\text{\coqemph[HC]{heap closure}}
  \\
  &\hskip5mm\N{\HE}~:=~\opt(\HC\times\HA)
  &\text{\coqemph[HE]{heap environment}}
  \\
  \N{\M{get}}
  &~:~\Heap\to\HA\to\opt(\HE)
\end{align*}
We will use the notation $\N{H[a]}:=\M{get}\,H\,a$.
We fix the semantics of heaps with an
inductive relation \coqemph[representsEnv]{$a\gg_HE$} 
relating heap addresses with environments:
\begin{mathpar}
  \inferrule*
  {H[a]=\some{\none}}
  {a\gg_H\nil}
  \and
  \inferrule*
  {H[a]=\some{\some{((p,b),c)}}\\
    p\gg_CP\\
    b\gg_HF\\
    c\gg_HE}
  {a\gg_H(P/F)::E}
\end{mathpar}

\begin{fact}[][representsEnv_functional]\label{fact-heap-rep-fun}
  The relation $a\gg_H E$ is functional.
\end{fact}

We also need an operation
  $\N{\M{put}}~:~\Heap\to\HC\to\HA\to\Heap\times\HA$
extending a heap with an environment.
Note that $\M{put}$ yields the extended heap
and the address of the extending environment.
We use the notation
\begin{align*}
  \N{H\incl H'}&~:=~\forall a.~H[a]\neq\eset~\to~H[a]=H'[a]
\end{align*}
to say that $H'$ is an \emph{extension} of $H$.
We fix the semantics of $\M{put}$
with the following requirement:
\begin{description}
\item[{\coqemph[HR1]{HR}}]
  If $\M{put}\,H\,g\,a=(H',b)$,
  then $H'[b]=(g,a)$ 
  and $H\incl H'$.
\end{description}

\begin{fact}[][representsEnv_extend]\label{fact-heap-ext}
  If $H\incl H'$ and $a\gg_HE$, 
  then $a\gg_{H'}E$.
\end{fact}

We define a relation \coqemph[representsClos]{$g\gg_H e$}
relating heap closures with proper closures:
\begin{align*}
  \coqlink[representsClos]{\N{(p,a)\gg_H (P,E)}}&~:=~p\gg_C P\land a\gg_H E
\end{align*}

\begin{fact}[][representsClos_extend]\label{fact-heap-ext-closure}
   If $H\incl H'$ and $g\gg_{H} e$, 
   then $g\gg_{H'} e$.
\end{fact}
\makeproof{fact-heap-ext-closure}{
  Follows with Fact~\ref{fact-heap-ext}.
}

We define a 
\coqemph[lookup]{lookup function} $\N{H[a,n]}:\opt(\HC)$ 
yielding the heap closure appearing at position~$n$
of the heap environment designated by $a$ in $H$:
\begin{align*}
  H[a,0]&~:=~\some{(p,b)}
  &&\text{if}~H[a]:=\some{((p,b),c)}
  \\
  H[a,\natS\,n]&~:=~H[c,n]
  &&\text{if}~H[a]:=\some{((p,b),c)}
\end{align*}

\begin{fact}[][lookup_unlinedEnv]\label{fact-heap-lookup}
  Let $a\gg_H E$.  Then:
  \begin{enumerate}
  \coqitem[nth_error_unlinedEnv] If $E[n]=\some e$, 
    then $H[a,n]=\some g$ and $g\gg_H e$
    for some $g$.
  \coqitem[lookup_unlinedEnv] If $H[a,n]=\some g$, 
    then $E[n]=\some e$ and $g\gg_H e$
    for some $e$.
  \end{enumerate}
\end{fact}
\makeproof{fact-heap-lookup}{
  Both claims follow by induction on $n$.
}

Here is one possible implementation of heaps:
\begin{align*}
  \HA&~:=~\nat\\
  \Heap&~:=~\List(\HC\times\HA)\\
  \M{get}~H\,0&~:=~\some\none\\
  \M{get}~H\,(\natS\,n)&~:=~\some{\some{(g,a)}}
    &\hskip-20mm\text{if}~H[n]=\some{(g,a)}\\
  \M{put}~H\,g\,a&~:=~(H\con[(g,a)],~\natS\,|H|)
\end{align*}
Note that with this implementation the address 0
represents the empty environment in every heap.

Given that Coq admits only structurally recursive functions,
writing a function computing $a\gg_HE$ is not straightforward.
The problem goes away if we switch to a step-indexed
function computing $a\gg_HE$.

\section{Heap Machine}
\setCoqFilename{M_heap}

The heap machine refines the closure machine
by representing programs as addresses into a fixed code
and environments as addresses into heaps that 
reside as additional component in the states
of the heap machine.

We assume a code structure 
providing types $\Code$ and $\PA$,
a code $C:\Code$,
and a heap structure 
providing types $\Heap$ and $\HA$.
\coqemph[state]{States of the heap machine} are triples
\begin{align*}
  \N{(T,V,H)}&~:~\List(\HC)\times\List(\HC)\times\Heap
\end{align*}
consisting of a control stack, an argument stack, and a heap.
The \coqemph[stepH]{reduction rules of the heap machine}
appear in Figure~\ref{fig:heap-red}.
They refine the reduction rules of the closure machine 
as one would expect.

\begin{figure}[t]
  \begin{align*}
    (p,a)::T,~V,~H
    &~\tred~T,\,V,\,H
    &&\text{if}~C[p]=\some{\ret}
    \\
    (p,a)::T,~V,~H
    &~\tred~(\#p,a)::T,~g::V,~H
    &&\text{if}~C[p]=\some{\var\,n}
    \\&&&\text{and}~H[a,n]=\some{g}
    \\
    (p,a)::T,~V,~H
    &~\tred~(\#p,a)::T,~(q,a)::V,~H
    &&\text{if}~C[p]=\some{\lamb\,q}
    \\
    (p,a)::T,~g::(q,b)::V,~H
    &~\bred~(q,c)::(\#p,a)::T,~V,~H'
    &&\text{if}~C[p]=\some{\app}
    \\&&&\text{and}~\M{put}\,H\,g\,b=\some{(H',c)}
  \end{align*}
  \vspace{-9mm} 
  \caption{Reduction rules of the heap machine}
  \label{fig:heap-red}
\end{figure}

Note that the application rule is
the only rule that allocates
new environments on the heap.
This is at first surprising 
since with practical machines (e.g., FAM and ZINC) 
heap allocation takes place 
when lambda commands are executed.
The naive allocation policy of 
our heap machine is a consequence
of the naive realisation of the lambda command
in the closure machine, 
which is common in formalisations of the SECD machine.
Given our refinement approach,
smart closure allocation would be prepared 
at the level of the naive stack machine
with programs that have explicit commands for
accessing and constructing closure environments.

Proving correctness of the heap machine is straightforward:

\begin{theorem}[Heap Machine to Closure Machine][heap_clos_refinement]
  \label{theo-heap-closure}
  Let a code structure, a code $C$, 
  and a heap structure be fixed.
  Let $T\gg_H\dot T$ and $V\gg_H\dot V$
  denote the pointwise extension of $g\gg_HE$ to lists.
  Then the relation
  \begin{align*}
    (T,V,H)\gg(\dot T,\dot V)
    &~:=~T\gg_H\dot T~\land~V\gg_H\dot V
  \end{align*}
  is a functional refinement.
  Moreover, 
  $([(p,a)],\nil,H)\gg([P/\nil],\nil)$
  for all $p$, $a$, $H$, and $P$ such that
  $p\gg_C P$ and $a\gg_H\nil$.
\end{theorem}  
\begin{proof}
  Follows with Facts~\ref{fact-code-rep-fun},
  \ref{fact-heap-rep-fun},
  \ref{fact-heap-ext},
  \ref{fact-heap-ext-closure},
  and~\ref{fact-heap-lookup}.
  Straightforward.
  \qed\end{proof}

Using the refinement from the closure machine to L, Theorem~\ref{theo-heap-closure}
and Fact~\ref{fact-ref-comp} we obtain a refinement from the Heap Machine
to L.
If
we instantiate the heap machine with
the realisation of codes from Section~\ref{sec:codes}
and the realisation of heaps from Section~\ref{sec:heaps}
we obtain a function compiling closed terms into initial states.
Moreover, given a function computing $a\gg_HE$,
we can obtain a decompiler 
for the states of the heap machine.
\section{Final Remarks}

The tail call optimisation can
be realised in our machines and 
accommodated in our verifications.
For this subprograms $\app;\ret$ are
executed such that no trivial continuation
(i.e., program $\ret$) is pushed on the control stack.

The control stack may be merged with the argument stack.
If this is done with explicit frames as in the SECD machine,
adapting our verification should be straightforward.
There is also the possibility to leave frames implicit
as in the modern SECD machine.
This will require different decompilation functions
and concomitant changes in the verification.

We could also switch to a 
$\lambda$-calculus with full substitution.
This complicates the definition of substitution
and the basic substitution lemmas but has the
pleasant consequence that we can drop 
the closedness constraints coming with
the correctness theorems  for the closure and heap machines.
The insight here is that a closure machine 
implements full substitution.
With full substitution we may reduce
\text{$\beta$-redexes} where the argument is a variable
and show a substitutivity property for small-step reduction.\todo{cleanup citations?}

\bibliography{bib}

\begin{thebibliography}{10}

\bibitem{Accattoli14}
B.~Accattoli, P.~Barenbaum, and D.~Mazza.
\newblock Distilling abstract machines.
\newblock In {\em Proceedings of the 19th {ACM} {SIGPLAN} International
  Conference on Functional Programming}, pages 363--376, 2014.

\bibitem{Baier2008}
C.~Baier and J.-P. Katoen.
\newblock {\em Principles of model checking}.
\newblock MIT Press, 2008.

\bibitem{Biernacka17}
M.~Biernacka, W.~Charatonik, and K.~Zielinska.
\newblock Generalized refocusing: From hybrid strategies to abstract machines.
\newblock In {\em LIPIcs}, volume~84. Schloss Dagstuhl-Leibniz-Zentrum f\"ur
  Informatik, 2017.

\bibitem{Biernacka07}
M.~Biernacka and O.~Danvy.
\newblock A concrete framework for environment machines.
\newblock {\em ACM Transactions on Computational Logic (TOCL)}, 9(1):6, 2007.

\bibitem{Cardelli84}
L.~Cardelli.
\newblock Compiling a functional language.
\newblock In {\em Proceedings of the 1984 ACM Symposium on LISP and Functional
  Programming}, pages 208--217. ACM, 1984.

\bibitem{Cousineau87}
G.~Cousineau, P.-L. Curien, and M.~Mauny.
\newblock The categorical abstract machine.
\newblock {\em Science of Computer Programming}, 8(2):173--202, 1987.

\bibitem{Cregut1990}
P.~Cr{\'e}gut.
\newblock An abstract machine for lambda-terms normalization.
\newblock In {\em Proceedings of the 1990 ACM Conference on LISP and Functional
  Programming}, pages 333--340, 1990.

\bibitem{DalLagoMartini08}
U.~{Dal Lago} and S.~Martini.
\newblock The weak lambda calculus as a reasonable machine.
\newblock {\em Theor. Comput. Sci.}, 398(1-3):32--50, 2008.

\bibitem{Danvy04refocusing}
O.~Danvy and L.~R. Nielsen.
\newblock Refocusing in reduction semantics.
\newblock {\em BRICS Report Series}, 11(26), 2004.

\bibitem{Felleisen86}
M.~Felleisen and D.~P. Friedman.
\newblock {\em Control Operators, the SECD-machine, and the
  $\lambda$-calculus}.
\newblock Indiana University, Computer Science Department, 1986.

\bibitem{LOLA}
Y.~Forster, F.~Kunze, and M.~Roth.
\newblock The strong invariance thesis for a $\lambda$-calculus.
\newblock {\em Workshop on Syntax and Semantics of Low-Level Languages (LOLA)},
  2017.

\bibitem{Forster17}
Y.~Forster and G.~Smolka.
\newblock Weak call-by-value lambda calculus as a model of computation in
  {C}oq.
\newblock In {\em Interactive Theorem Proving - 8th International Conference},
  pages 189--206. Springer, LNCS 10499, 2017.

\bibitem{Hardin98}
T.~Hardin, L.~Maranget, and B.~Pagano.
\newblock Functional runtime systems within the lambda-sigma calculus.
\newblock {\em Journal of Functional Programming}, 8(2):131--176, 1998.

\bibitem{Landin64}
P.~J. Landin.
\newblock The mechanical evaluation of expressions.
\newblock {\em The Computer Journal}, 6(4):308--320, 1964.

\bibitem{leroy90}
X.~Leroy.
\newblock The {ZINC} experiment: an economical implementation of the {ML}
  language.
\newblock Technical report, INRIA, 1990.

\bibitem{Leroy16}
X.~Leroy.
\newblock Functional programming languages, {Part II}: Abstract machines, the
  {M}odern {SECD}.
\newblock Lectures on Functional Programming and Type Systems, MPRI course 2-4,
  slides and Coq developments, \url{https://xavierleroy.org/mpri/2-4/}, 2016.

\bibitem{Leroy2009}
X.~Leroy and H.~Grall.
\newblock Coinductive big-step operational semantics.
\newblock {\em Information and Computation}, 207(2):284--304, 2009.

\bibitem{Plotkin75}
G.~D. Plotkin.
\newblock Call-by-name, call-by-value and the $\lambda$-calculus.
\newblock {\em Theor. Comput. Sci.}, 1(2):125--159, 1975.

\bibitem{Ramsdell99}
J.~D. Ramsdell.
\newblock The tail-recursive {SECD} machine.
\newblock {\em Journal of Automated Reasoning}, 23(1):43--62, 1999.

\bibitem{rittri1988}
M.~Rittri.
\newblock {\em Proving the correctness of a virtual machine by a bisimulation}.
\newblock Chalmers University and University of G\"oteborg, 1988.
\newblock Licentiate thesis.

\bibitem{Swierstra12}
W.~Swierstra.
\newblock From mathematics to abstract machine: {A} formal derivation of an
  executable {Krivine} machine.
\newblock In {\em Proceedings Fourth Workshop on Mathematically Structured
  Functional Programming}, pages 163--177, 2012.

\bibitem{Coq}
{The Coq Proof Assistant}.
\newblock \url{http://coq.inria.fr}, 2018.

\end{thebibliography}
\bibliographystyle{abbrv}

\end{document}